\documentclass{article}

\def\noheaderplainsetup{

\topmargin=0pt \headheight=0pt \headsep=0pt  \oddsidemargin=0pt \evensidemargin=0pt  \textheight=9.1truein \textwidth=6.5truein}   

\noheaderplainsetup

\usepackage{amsfonts}

\begin{document}


\newcommand{\cltw}{\mbox{LKg}}


\newcommand{\code}[1]{\ulcorner #1 \urcorner}
\newcommand{\pintimpl}{\mbox{\hspace{2pt}\raisebox{0.033cm}{\tiny $>$}\hspace{-0.18cm} \raisebox{-0.043cm}{\large --}\hspace{2pt}}} 
\newcommand{\st}{\mbox{\raisebox{-0.05cm}{$\circ$}\hspace{-0.13cm}\raisebox{0.16cm}{\tiny $\mid$}\hspace{2pt}}}  
\newcommand{\sti}{\mbox{\raisebox{-0.02cm}
{\scriptsize $\circ$}\hspace{-0.121cm}\raisebox{0.08cm}{\tiny $.$}\hspace{-0.079cm}\raisebox{0.10cm}
{\tiny $.$}\hspace{-0.079cm}\raisebox{0.12cm}{\tiny $.$}\hspace{-0.085cm}\raisebox{0.14cm}
{\tiny $.$}\hspace{-0.079cm}\raisebox{0.16cm}{\tiny $.$}\hspace{1pt}}}
\newcommand{\intimpl}{\mbox{\hspace{2pt}$\circ$\hspace{-0.14cm} \raisebox{-0.043cm}{\Large --}\hspace{2pt}}} 
\newcommand{\fintimpl}{\mbox{\hspace{2pt}$\bullet$\hspace{-0.14cm} \raisebox{-0.058cm}{\Large --}\hspace{-6pt}\raisebox{0.008cm}{\scriptsize $\wr$}\hspace{-1pt}\raisebox{0.008cm}{\scriptsize $\wr$}\hspace{4pt}}} 

\newcommand{\adi}{\hspace{2pt}\raisebox{0.02cm}{\mbox{\small $\sqsupset$}}\hspace{2pt}} 
\newcommand{\plus}{\mbox{\hspace{1pt}\raisebox{0.05cm}{\tiny\boldmath $+$}\hspace{1pt}}}
\newcommand{\minus}{\mbox{\hspace{1pt}\raisebox{0.05cm}{\tiny\boldmath $-$}\hspace{1pt}}}
\newcommand{\mult}{\mbox{\hspace{1pt}\raisebox{0.05cm}{\tiny\boldmath $\times$}\hspace{1pt}}}
\newcommand{\equals}{\mbox{\hspace{1pt}\raisebox{0.05cm}{\tiny\boldmath $=$}\hspace{1pt}}}
\newcommand{\notequals}{\mbox{\hspace{1pt}\raisebox{0.05cm}{\tiny\boldmath $\not=$}\hspace{1pt}}}
\newcommand{\successor}{\mbox{\hspace{1pt}\boldmath $'$}}

\newcommand{\elz}[1]{\mbox{$\parallel\hspace{-3pt} #1 \hspace{-3pt}\parallel$}} 
\newcommand{\elzi}[1]{\mbox{\scriptsize $\parallel\hspace{-3pt} #1 \hspace{-3pt}\parallel$}}
\newcommand{\emptyrun}{\langle\rangle} 
\newcommand{\oo}{\bot}            
\newcommand{\pp}{\top}            
\newcommand{\xx}{\wp}               
\newcommand{\legal}[2]{\mbox{\bf Lr}^{#1}_{#2}} 
\newcommand{\win}[2]{\mbox{\bf Wn}^{#1}_{#2}} 
\newcommand{\seq}[1]{\langle #1 \rangle}           


\newcommand{\pst}{\mbox{\raisebox{-0.01cm}{\scriptsize $\wedge$}\hspace{-4pt}\raisebox{0.16cm}{\tiny $\mid$}\hspace{2pt}}}
\newcommand{\pcost}{\mbox{\raisebox{0.12cm}{\scriptsize $\vee$}\hspace{-4pt}\raisebox{0.02cm}{\tiny $\mid$}\hspace{2pt}}}

\newcommand{\gneg}{\mbox{\small $\neg$}}                  
\newcommand{\mli}{\hspace{2pt}\mbox{\small $\rightarrow$}\hspace{2pt}}                      
\newcommand{\cla}{\mbox{$\forall$}}      
\newcommand{\cle}{\mbox{$\exists$}}        
\newcommand{\mld}{\hspace{2pt}\mbox{\small $\vee$}\hspace{2pt}}     
\newcommand{\mlc}{\hspace{2pt}\mbox{\small $\wedge$}\hspace{2pt}}   
\newcommand{\mlci}{\hspace{2pt}\mbox{\footnotesize $\wedge$}\hspace{2pt}}   
\newcommand{\ade}{\mbox{\large $\sqcup$}}      
\newcommand{\ada}{\mbox{\large $\sqcap$}}      
\newcommand{\add}{\hspace{2pt}\mbox{\small $\sqcup$}\hspace{2pt}}                     
\newcommand{\adc}{\hspace{2pt}\mbox{\small $\sqcap$}\hspace{2pt}} 
\newcommand{\adci}{\hspace{2pt}\mbox{\footnotesize $\sqcap$}\hspace{2pt}}              
\newcommand{\clai}{\forall}     
\newcommand{\clei}{\exists}        
\newcommand{\tlg}{\bot}               
\newcommand{\twg}{\top}               
\newcommand{\col}[1]{\mbox{$#1$:}}


\newtheorem{theoremm}{Theorem}[section]
\newtheorem{factt}[theoremm]{Fact}
\newtheorem{corollaryy}[theoremm]{Corollary}
\newtheorem{definitionn}[theoremm]{Definition}
\newtheorem{thesiss}[theoremm]{Thesis}
\newtheorem{lemmaa}[theoremm]{Lemma}
\newtheorem{conventionn}[theoremm]{Convention}
\newtheorem{examplee}[theoremm]{Example}
\newtheorem{exercisee}[theoremm]{Exercise}
\newtheorem{remarkk}[theoremm]{Remark}
\newenvironment{definition}{\begin{definitionn} \em}{ \end{definitionn}}
\newenvironment{thesis}{\begin{thesiss} \em}{ \end{thesiss}}
\newenvironment{theorem}{\begin{theoremm}}{\end{theoremm}}
\newenvironment{lemma}{\begin{lemmaa}}{\end{lemmaa}}
\newenvironment{fact}{\begin{factt}}{\end{factt}}
\newenvironment{corollary}{\begin{corollaryy}}{\end{corollaryy}}
\newenvironment{convention}{\begin{conventionn} \em}{\end{conventionn}}
\newenvironment{example}{\begin{examplee} \em}{\end{examplee}}
\newenvironment{exercise}{\begin{exercisee} \em}{\end{exercisee}}
\newenvironment{remark}{\begin{remarkk} \em}{\end{remarkk}}
\newenvironment{proof}{ {\bf Proof.} }{\  \rule{2.5mm}{2.5mm} \vspace{.2in} }

\title{A Heuristic Proof Procedure for First-Order Logic}
\author{Keehang Kwon  \\ 
  {\small Department of Computing Sciences, DongA University, South Korea.
 khkwon@dau.ac.kr}}
\date{}
\maketitle

\begin{abstract}
Inspired by the efficient proof procedures discussed in {\em Computability logic} \cite{Jap03,Japic,Japfin}, we describe a heuristic proof procedure for
first-order logic. This is a variant of Gentzen sequent system and has the following features:
(a)~ it views sequents as games between the machine and the environment, and
(b)~ it views proofs as a winning strategy of the machine.

From this game-based viewpoint, a poweful heuristic can be extracted and
a fair  degree of determinism in proof search can be obtained.
This article proposes a new  deductive system \cltw\ with respect to 
first-order logic and proves its soundness and completeness.
\end{abstract}

\



\section{Introduction}\label{intr}

The Gentzen sequent system LK plays a key role in modern theorem proving.
Unfortunately, the LK system and its variants such as $focused$ LK
(as well as resolution and tableux
(see \cite{Reeves} for discussions))
are typically based on blind search and, therefore,
does not provide the best strategy if we want a short proof.

In this paper, inspired by the seminal
work of \cite{Jap03}, we present a variant of LK, called \cltw\ (g for game),  which yields a
proof in normal form with the following features:

\begin{itemize}

\item All the quantifier inferences are processed first.
This is achieved via deep inference.

\item If there are several quantifiers to resolve in the sequent,
we apply to sequents a technique  called {\it stability analysis},
a powerful heuristic  technique which greatly cuts down the search space
for finding a proof.

\end{itemize}

In essence, \cltw\ is a $game$-$viewed$ proof which 
captures  $game$-$playing$ nature  in proof search. It views 

\begin{enumerate}

\item sequents as games between the machine and the environment,

\item proofs as a winning strategy of the machine, and 

\item $\cla x F$  as the env's move and $\cle x F$ as the machine's move.

\end{enumerate}

At each stage, we construct a proof by the following rules:

\begin{enumerate}

\item If the sequent is stable, then it means that  the machine is the current winner.
      In this case, it requests the user to make a move.

\item If the sequent is instable, then it means that the environment is the current winner.
      In this case,  the machine makes a move.

\end{enumerate}
 
In this way, a fair (probably maximum) degree of determinism can be obtained from the \cltw\
proof system.

In this paper we present the proof procedure for first-order classical logic.
The remainder of this paper is structured as follows. We describe \cltw\ in
the next section. In Section \ref{examples}, we
present some examples of derivations. In Section 4, we prove the soundness
and completeness of \cltw.
Section 5 concludes the paper.

\section{The  logic $\cltw$}\label{ss6}

 The  formulas  are   the standard first-order classical formulas,
 with the features that
 (a) $\twg,\tlg$ are added, and (b) $\gneg$ is only allowed to be applied to atomic formulas.
Thus we assume that formulas are in negation normal form.

The  deductive system $\cltw$ below axiomatizes the set of valid formulas. 
$\cltw$ is a one-sided sequent calculus system, where 
a  sequent is a multiset of formulas.
Our presentation closely follows the one in \cite{Jap03}.

First, we need to define some terminology.

\begin{enumerate}
\item A {\bf surface occurrence} of a subformula is an occurrence that is 
not in the scope of any quantifiers ($\cla$ and/or $\cle$). 
\item  A sequent is {\bf propositional} iff all of its formulas are so. 
\item The {\bf propositionalization} $\elz{F}$ of a formula $F$ is the result of replacing
  in $F$ all $\cle$-subformulas by $\tlg$, and all  $\cla$-subformulas by $\twg$.
   The {\bf propositionalization} $\elz{F_1,\ldots,F_n}$ of a sequent 
$F_1,\ldots,F_n$ is the propositional formula $\elz{F_1}\mld\ldots\mld \elz{F_n}.$
\item A sequent  is said to be {\bf stable} iff its propositionalization is classically valid; otherwise it is {\bf unstable}.
  
\item The notation $F[E]$ repesents  a formula $F$ together with some
  surface occurrence of a subformula $E$. 
 
\end{enumerate}

\begin{center}
\begin{picture}(100,30)

\put(0,10){\bf THE RULES OF $\cltw$}

\end{picture}
\end{center}

$\cltw$ has the five rules listed below, with the following additional conditions: 
\begin{enumerate}
\item $X$:stable means that $X$ must meet the condition that it is
  stable. Similarly for $X$:unstable.
\item $\Gamma$ is a multiset of formulas and $F$ is a formula.
\item In $\cle$-Choose,  $t$ is  a closed term,  and $H(t)$ is the result of replacing by $t$ all free occurrences of $x$ in $H(x)$.

\end{enumerate}

\begin{center}
\begin{picture}(74,70)
\put(12,50){\bf  Fail}
\put(20,30){$\tlg$}
\put(0,22){\line(1,0){45}}
\put(55,20){( $X$ has no surface  occurrences of  $\cle x H(x)$)}
\put(8,8){$X$:unstable}
\end{picture}
\end{center}

\begin{center}
\begin{picture}(74,70)

\put(12,50){\bf $\cle$-Choose}
\put(8,30){$\Gamma,F[ H(t)]$}
\put(0,22){\line(1,0){78}}
\put(8,8){$\Gamma,F[\cle x H(x)]$:unstable}

\end{picture}
\end{center}

\begin{center}
\begin{picture}(74,70)

\put(12,50){\bf Replicate}
\put(8,8){$\Gamma,F[\cle x H(x)]$:unstable}
\put(0,22){\line(1,0){83}}
\put(0,30){$\Gamma,F[\cle x H(x)],F[\cle x H(x)]$}
\end{picture}
\end{center}

\begin{center}
\begin{picture}(70,70)
\put(12,50){\bf Succ}
\put(8,30){$\twg$}
\put(0,22){\line(1,0){45}}
\put(55,20){ ( $X$ has no surface occurrences of  $\cla xH(x)$)}
\put(8,8){$X$:stable}
\end{picture}
\end{center}

\begin{center}
\begin{picture}(74,70)
\put(20,50){\bf $\cla$-Choose }
\put(20,30){$\Gamma, F[H(\alpha)]$}
\put(0,22){\line(1,0){85}}
\put(85,20){   ($\alpha$ is a new constant)}
\put(20,8){$\Gamma, F[\cla xH(x)]$: stable}
\end{picture}
\end{center}

In the above,  the ``Replicate''  rule is an optimized version of
what is  known as Contraction, where contraction occurs only when
there is a surface occurrence of $\cle xH(x)$.   

A {\bf $\cltw$-proof} of a sequent $X$ is a sequence $X_1,\ldots,X_n$ of sequents, with $X_n=X$,
$X_1 = \twg$ such that, each $X_i$ follows  by one of the rules of $\cltw$ from $X_{i-1}$.

\section{Examples}\label{examples}

Below we describe some examples.

\begin{example}\label{j28b}
The formula $\cla x\cle y\bigl(p(x)\mli p(y)\bigr)$ is provable in $\cltw$ as follows:\vspace{7pt}

\noindent 1. $\begin{array}{l}
p(\alpha)\mli p(\alpha)
\end{array}$  \ \ Succ\vspace{3pt}

\noindent 2. $\begin{array}{l}
\cle y\bigl(p(\alpha)\mli p(y)\bigr)
\end{array}$  \ \ $\cle$-Choose\vspace{3pt}

\noindent 3. $\begin{array}{l}
\cla x\cle y\bigl(p(x)\mli p(y)\bigr)
\end{array}$  \ \ $\cla$-Choose\vspace{7pt}

\end{example}

\begin{example}\label{j28b}
The formula $\cle y\cla x\bigl(p(x)\mli p(y)\bigr)$ is provable in $\cltw$ as follows:\vspace{7pt}

\noindent 1. $\begin{array}{l}
\bigl(p(\alpha_1)\mli p(a)\bigr),\bigl(p(\alpha_2)\mli p(\alpha_1)\bigr)
\end{array}$  \ \ Succ\vspace{3pt}

\noindent 2. $\begin{array}{l}
\bigl(p(\alpha_1)\mli p(a)\bigr), \cla x\bigl(p(x)\mli p(\alpha_1)\bigr)
\end{array}$  \ \ $\cla$-Choose\vspace{3pt}

\noindent 3. $\begin{array}{l}
\bigl(p(\alpha_1)\mli p(a)\bigr), \cle y\cla x\bigl(p(x)\mli p(y)\bigr)
\end{array}$  \ \ $\cle$-Choose\vspace{3pt}

\noindent 4. $\begin{array}{l}
\cla x\bigl(p(x)\mli p(a)\bigr), \cle y\cla x\bigl(p(x)\mli p(y)\bigr)
\end{array}$  \ \ $\cla$-Choose\vspace{3pt}

\noindent 5. $\begin{array}{l}
\cle y\cla x\bigl(p(x)\mli p(y)\bigr), \cle y\cla x\bigl(p(x)\mli p(y)\bigr)
\end{array}$  \ \ $\cle$-Choose\vspace{3pt}

\noindent 6. $\begin{array}{l}
\cle y\cla x\bigl(p(x)\mli p(y)\bigr)
\end{array}$  \ \ Replicate\vspace{7pt}

On the other hand, the formula $ \bigl(\cle x p(x)\mli \cla y p(y)\bigr)$ which is invalid
can be seen to be unprovable.
This can be derived only by two $\cla$-Choose rules  and then the premise should be of the form
$\neg p(\alpha_1), p(\alpha_2)$ for some new constants $\alpha_1,\alpha_2$. The latter is not classically valid.
\end{example}

\section{The soundness and completeness of \cltw}\label{ssc}

We now present the soundness and completeness of \cltw.

\begin{theorem}\label{main}
  
  \begin{enumerate}

  \item If  $\cltw$ terminates with success for $X$, then $X$ is valid.

  \item If  $\cltw$ terminates with failure  for $X$, then $X$ is invalid.

    \item If  $\cltw$ does not terminate  for $X$, then $X$ is invalid.

    \end{enumerate}
\end{theorem}

\begin{proof}  Consider an arbitrary sequent $X$. 
\vspace{5pt}

{\em Soundness:} Induction on the length of derivatons.

{\em Case 1:} $X$ is derived from $Y$ by $\cle$-Choose.   By the induction hypothesis,
$Y$ is valid, which implies that  $X$ is valid.

{\em Case 2}:   $X$ is derived from $Y$ by Replicate. By the induction hypothesis,  $Y$ is
valid. Then, it is easy to see that $X$ is valid.

{\em Case 3}:  $X$ is derived from $Y$ by Succ.

  In this case, we know that there is no surface occurrences of $\cla x H(x)$ in $X$ and 
  $\elz{X}$ is classically valid. It is then
  easy to see that, reversing the propositionalization
  of $\elz{X}$ (replacing $\tlg$ by any formula of the form
  $F[\cle x H(x)]$)
  preserves  validity. For example, if $X$ is $p(a) \mli p(a), \cle x q(x)$,
  then $\elz{X}$ is valid and $X$ is valid as well.

  {\em Case 4}:  $X$ is derived from $Y$ by $\cla$-Choose.

  Thus, there is an occurrence of $\cla x H(x)$ in $X$.
The machine makes a move by picking up  some fresh constant $c$ not occurring in $X$.
Then, by the induction hypothesis, the premise is valid.
Now consider any interpretation $I$ that makes the premise true.
Then it is easy to see that the conclusion is true in $I$.
It is commonly known as ``generalization on constants''.

 \vspace{5pt}

{\em Completeness:}  Assume $\cltw$ terminates with failure. 
 
We proceed by induction on the length of derivations.

If $X$ is stable, then there should be a $\cltw$-unprovable sequent $Y$ with the following
condition.

 {\em Case 1: $\cla$-Choose:} $X$ has the form $\Gamma,F[\cla xG(x)]$, and
 $Y$ is $\Gamma,F[G(\alpha)]$, where $\alpha$ is a new constant not occurring in $X$.
 In this case, $Y$ is a $\cltw$-unprovable sequent, for otherwise $X$ is $\cltw$-provable.
 By the induction hypothesis, $Y$ is not true in some interpretation $I$. Then it is easy to
 see that $X$ is not true in  $I$. Therefore $X$ is not valid.

Next, we consider the cases when $X$ is not stable. Then there are three cases to consider.

{\em Case 2.1: Fail}: In this case,
there is no surface occurrence of $\cle x G(x)$ and the alorithm terminates with fauilure.
As $X$ is not stable, $\elz{X}$ is not classically valid. If we reverse the propositionalization
of $\elz{X}$ by replacing $\twg$ by any formula with some surface occurrence of $\cla G(x)$,
we observe that invalidity is preserved. Therefore, $X$ is not valid.

{\em Case 2.2: $\cle$-Choose}: In this case,  $X$ has the form $\Gamma,F[\cle x e G(x)]$
, and $Y(t)$ is $\Gamma,F[G(t)]$, where $t$ is a closed term.
In this case, $Y(t)$ is a $\cltw$-unprovable sequent for any $t$,
for otherwise $X$ is $\cltw$-provable. 
By the induction hypothesis, none of $Y(t)$ is  valid and thus none of $Y(t)$ is
not true in some interpretation $I$. Then it is easy to
 see that $X$ is not true in  $I$. Therefore $X$ is not valid.

 {\em Case 2.3: Replicate}: In this case,  $X$ has the form
 $\Gamma,F[\cle x  G(x)]$
, and $Y$ is $\Gamma,F[\cle x G(x)],F[\cle x  G(x)]$.
In this case, $Y$ is a $\cltw$-unprovable sequent,
for otherwise $X$ is $\cltw$-provable. 
By the induction hypothesis, $Y$ is  not valid and  is
not true in some interpretation $I$. Then it is easy to
 see that $X$ is not true in  $I$. Therefore $X$ is not valid.

 Now assume \cltw\ is not terminating, because Replicate occurs infinitely many times.
 We prove this by contradition.

 Assume that $X$ is valid but unprovable. Let $Z$ be an infinite multiset of propositional formulas
 obtained by applying infinite numbers of Replicate, together with $\cle$-Choose and $\cla$-Choose
 rules. 
 Then it is easy to see that $Z$ remains still valid but unprovable.
 By the compactness theorem on propositional logic, there is a finite subset $Z'$ of $Z$,  which is a valid sequent.
 Then there must be a step $t$ in the procedure such that, after $t$,  $Z'$ is derived.
 Then it is easy to see that  $Z'$ remains \cltw-unprovable. However, as $Z'$ is valid, it
 must be provable by Succ. This is a contradiction and, therefore,
 $X$ is not valid.

\end{proof}

\section{Some Optimizations}\label{ss6}

Although $\cltw$ performs well for valid sequents, it performs poorly
for invalid sequents. For example, it does not even terminate
for the invalid sequent $p(a), p(b) \mlc \cle x q(x)$.

For this reason, what we need is
a good heuristic for determining, in a simple yet effective way,
whether a sequent is invalid. In this section, we employ a simple
heuristic
called {\it maximum propositionalization} which replaces
in a sequent $X$ all $\cle$-subformulas by $\twg$. If $X'$ is obtained
from $X$ by maximum propositionalization, then it is easy to observe that
if $X'$ is invalid, then $X$ is invalid as well.

The  deductive system $\cltw'$ uses this heuristic. 
First, we need to define some terminology.

\begin{enumerate}

\item The {\bf max-propositionalization} $\elz{F}_{max}$ of a formula $F$ is the result of replacing
  in $F$ all $\cle$-subformulas by $\twg$, and all  $\cla$-subformulas by $\twg$. This process naturally extends to sequents.
  
\item  The {\bf min-propositionalization} $\elz{F}_{min}$ of a formula $F$ is the result of replacing
  in $F$ all $\cle$-subformulas by $\tlg$, and all  $\cla$-subformulas by $\tlg$. This process naturally extends to sequents.

 \item A sequent  is said to be {\bf max-p-invalid} iff its max-propositionalization is classically invalid.
   A sequent  is said to be {\bf min-p-valid} iff its min-propositionalization is classically valid.

\end{enumerate}

\begin{center}
\begin{picture}(100,30)

\put(0,10){\bf THE RULES OF $\cltw'$}

\end{picture}
\end{center}

Below, $X$:stable means that $X$ is stable but not min-p-valid.
  stable. Similarly  $X$:unstable means that $X$ is unstable
  but not min-p-valid.

  \begin{center}
\begin{picture}(74,70)
\put(12,50){\bf  Fail}
\put(20,30){$\tlg$}
\put(0,22){\line(1,0){45}}
\put(8,8){$X$: max-p-invalid}
\end{picture}
\end{center}

\begin{center}
\begin{picture}(74,70)

\put(12,50){\bf $\cle$-Choose}
\put(8,30){$\Gamma,F[ H(t)]$}
\put(0,22){\line(1,0){78}}
\put(8,8){$\Gamma,F[\cle x H(x)]$:unstable}

\end{picture}
\end{center}

\begin{center}
\begin{picture}(74,70)

\put(12,50){\bf Replicate}
\put(8,8){$\Gamma,F[\cle x H(x)]$:unstable}
\put(0,22){\line(1,0){83}}
\put(0,30){$\Gamma,F[\cle x H(x)],F[\cle x H(x)]$}
\end{picture}
\end{center}

\begin{center}
\begin{picture}(70,70)
\put(12,50){\bf Succ}
\put(8,30){$\twg$}
\put(0,22){\line(1,0){45}}
\put(8,8){$X$: min-p-valid}
\end{picture}
\end{center}

\begin{center}
\begin{picture}(74,70)
\put(20,50){\bf $\cla$-Choose }
\put(20,30){$\Gamma, F[H(\alpha)]$}
\put(0,22){\line(1,0){85}}
\put(85,20){   ($\alpha$ is a new constant)}
\put(20,8){$\Gamma, F[\cla xH(x)]$: stable}
\end{picture}
\end{center}

The heuristic employed in $\cltw'$ is quite simple and needs to be
improved. For example, it does not apply well to  the invalid
sequent $p(a), \cle xp(x)$. It would be nice
to improve our heuristic so that it can apply to a wider class of
invalid sequents.

\end{document}